\documentclass[11pt,reqno]{article}
\usepackage[margin=1in]{geometry}
\usepackage{dsfont, amssymb,amsmath,amscd,latexsym, amsthm, amsxtra,amsfonts}
\allowdisplaybreaks[4]
\usepackage{lineno}
\usepackage[all]{xy}
\usepackage[active]{srcltx}
\usepackage{tikz}
\usepackage[round]{natbib}
\usepackage{bbm}
\usepackage{enumerate}
\usepackage{mathrsfs}
\usepackage{graphicx}
\usepackage{subcaption}
\usepackage{comment}
\usepackage{mathtools}
\usepackage{cases}
\usepackage{tcolorbox}
\tcbuselibrary{most}

\usetikzlibrary{calc,arrows}
\usepackage{verbatim}
\usepackage{color}
\usepackage{epstopdf}
\usepackage[affil-it]{authblk}
\usepackage{bm}
\usepackage[title]{appendix}

\usepackage{etoolbox}
\makeatletter
\patchcmd{\NAT@citex}
{\NAT@nmfmt{\NAT@nm} \NAT@date}
{\NAT@nmfmt{\NAT@nm}\linebreak[0] \NAT@date}
{}{}
\makeatother

\newtheorem{theorem}{Theorem}[section]

\newtheorem{definition}[theorem]{Definition}
\newtheorem{example}[theorem]{Example}

\newtheorem{lemma}[theorem]{Lemma}

\newtheorem{proposition}[theorem]{Proposition}
\newtheorem{remark}[theorem]{Remark}
\numberwithin{equation}{section}

\DeclareMathOperator*{\esssup}{ess\,sup\,}
\newcommand{\me}{\mathrm{e}}
\newcommand{\logn}{\mathrm{LogNormal}}

\newcommand{\md}{\mathrm{d}}

\newcommand{\mR}{\mathbb{R}}
\newcommand{\mU}{\mathbb{U}}

\newcommand{\mS}{\mathbb{S}}

\newcommand{\mE}{\mathbb{E}}
\newcommand{\mF}{\mathbb{F}}

\renewcommand{\epsilon}{\varepsilon}

\newcommand{\F}{\mathcal{F}}

\newcommand{\Ti}{\mathcal{T}}

\newcommand{\barpi}{\bar{\pi}}

\newcommand{\cred}{\color{red}}

  \usepackage[pdfstartview=FitH, bookmarksnumbered=true,bookmarksopen=true, colorlinks=true, pdfborder={0 0 1}, citecolor=blue, linkcolor=blue,urlcolor=blue]{hyperref}
\usepackage{graphics}
\graphicspath{{figures/}}

\title{An Integral Equation in Portfolio Selection with Time-Inconsistent Preferences
}

\author{Zongxia Liang$^a$\thanks{Email: \texttt{liangzongxia@mail.tsinghua.edu.cn}}\ \ \ \ \ Sheng Wang$^a$\thanks{Email: \texttt{wangs20@mails.tsinghua.edu.cn}} \ \ \ \ \ Jianming Xia$^b$\thanks{Email: \texttt{xia@amss.ac.cn}} \ \ \ \ \ 
}	
\affil{$^a$Department of Mathematical Sciences, Tsinghua University, China \\
$^b$RCSDS, NCMIS, Academy of Mathematics and Systems Science, Chinese Academy of Sciences, Beijing 100190, China}

\begin{document}
	
\maketitle
	
\begin{abstract}
     This paper discusses a nonlinear integral equation arising from    portfolio selection with a class of time-inconsistent preferences.  We propose a unified framework requiring minimal assumptions, such as right-continuity of market coefficients and square-integrability of the market price of risk. Our main contribution is proving the existence and uniqueness of the square-integrable solution for the integral equation under mild conditions.  Illustrative applications include the mean-variance portfolio selection and the utility maximization with random risk aversion.
\end{abstract}
	

 \section{Introduction}{\label{intro}}
 A time-inconsistent preference usually leads to the time-inconsistency of the optimal strategy:  the strategy that is optimal today may not remain optimal tomorrow.
 To address the issue of time inconsistency, the intra-personal equilibrium
 is originally introduced by
 \cite{Strotz1955}.
 Since a precise definition of continuous-time intra-personal equilibrium was first provided by \cite{ekeland2006being}, the continuous-time intra-personal equilibrium has been studied for various time-inconsistent preferences. For instance, the mean-variance (MV) portfolio selection problem is investigated by \cite{basak2010dynamic}, \cite*{hu2012time, Hu2017}, \cite*{bjork2014mean}, and \cite*{Dai2021}; the non-exponential discounting problem is examined by\\ \cite{ekeland2006being}, \cite{ekeland2008investment}, \cite{hamaguchi2021time}, and \cite{mbodji2023portfolio}. For a broader discussion on  continuous-time intra-personal equilibrium, see \cite{yan2019time}, \cite{he2021equilibrium}, \cite{Hernandez2023}, and the references therein. 
 
 The continuous-time intra-personal equilibria for a class of  time-inconsistent  preferences can be characterized by the solutions to a class of nonlinear integral (or differential) equations.
Below, we provide  four specific preferences to illustrate this. 
The first one is the MV preference with the risk aversion parameter being inversely proportional to the wealth level,
which is proposed by \cite*{bjork2014mean}. In a market with constant coefficients, the authors apply the Arzelà-Ascoli theorem to establish the existence and uniqueness of the solution to their integral equation (4.11), thereby confirm the existence and uniqueness of the equilibrium strategy among state-independent strategies. The second one is the preference with random risk aversion of \cite{desmettre2023equilibrium}, which addresses the problem of maximizing utility with random risk aversion. This preference is time-inconsistent due to the inclusion of nonlinear functions of expectation. The authors characterize the equilibrium strategies by an infinite-dimensional system of ordinary differential equations (ODEs). Specifically, when the risk aversion has two possible values, this system reduces to a three-dimensional ODE system. However, the existence and uniqueness of the solution to this system is not provided there. The third one
is the generalized disappointment aversion (GDA) preference, which is discussed in \cite*{liang2024dynamic}. Under GDA, the implicit definition of certainty equivalence introduces time inconsistency. The equilibrium strategies are characterized by a nonlinear integral equation.   The fourth one is the (endogenous habit, or mean) scaled MV preferences and (endogenous habit formation, or generalized) mean–standard-deviation preferences of \cite*{kryger2020optimal}. The equilibrium strategy is characterized by their ODE system $(21)-(25)$ (in the case of endogenous habit, scaled MV preference, the existence and uniqueness of the solution to  the ODE system remains unproved), which implies that the equilibrium strategy is absolutely continuous.

This paper presents a unified approach to discuss the integral equation arising from  portfolio selection  with a large class of time-inconsistent preferences and portfolio bounds. We establish the existence and the uniqueness of solution to this equation.

Our main contribution is a unified framework that minimizes assumptions on market parameters. Specifically, we require only that the market coefficients are right-continuous and the market price of risk is square-integrable. In contrast,  \cite*{bjork2014mean},\\ \cite*{kryger2020optimal}, and \cite{desmettre2023equilibrium} assume constant market coefficients: the methods in \cite*{bjork2014mean} and \citet*{kryger2020optimal} cannot handle right-continuous coefficients, while \cite{desmettre2023equilibrium} and a part of \citet*{kryger2020optimal}) do not establish the existence and uniqueness of the solution to their ODE system;  \cite*{liang2024dynamic} assume that the market coefficients are right-continuous and bounded.  

 The remainder of the paper is organized as follows. Section \ref{sec:integral:equation} investigates the existence and the uniqueness of the solution of a specific type of integral equation. Section \ref{sec:time-inconsistent} derives this integral equation from    portfolio selection  with a class of time-inconsistent preferences. Finally, Section \ref{sec:applications} provides two illustrative  applications.

\section{An Integral equation}\label{sec:integral:equation}

Let $T>0$ be a constant and $d\ge1$ be an integer. Given a fixed function $h: [0,T)\times[0,\infty)\times\mR\to\mR$ and a fixed $\mR^d$-valued function $\lambda\in L^2((0,T),\mR^d)$, we consider the following fully non-linear integral equation:\footnote{For a matrix $A$, $A^\top$ denotes the transpose of $A$. For a vector $z\in\mR^d$,  $|z|\triangleq\sqrt{z^\top z}$.
For any $0\leq t_1< t_2\leq T$, let $L^2((t_1,t_2),\mR^d)$ denote the set of $\mR^d$-valued Lebesgue measurable functions $\phi$ such that
$\Vert \phi\Vert_{L^2((t_1,t_2),\mR^d)}\triangleq \left(\int_{t_1}^{t_2}|\phi(s)|^2\md s\right)^{\frac12}<\infty$.}
\begin{align}\label{eq:integral:h}
a(t)=\mathcal{P}_t\left(h\left(t,\sqrt{\int_t^T|a(s)|^2\md s},\int_t^Ta^\top(s)\lambda(s)\md s\right)\lambda(t)\right), \quad t\in(0,T),
\end{align}
where $a:(0,T)\to\mR^d$ is a function to be determined by the equation and, for every $t\in[0,T)$, $\mathcal{P}_t$ is the projection from $\mR^d$ to a convex and closed set $\mathbb{U}_t\subseteq \mR^d$ with $\mathbf{0}\in\mathbb{U}_t$,  i.e.,
\begin{align}\label{eq:projection}
			\mathcal{P}_t(k)\triangleq \arg \min_{\tilde{k}\in\mU_t}|\tilde{k}-k|^2,\quad k\in\mR^d.
\end{align}

\begin{definition}
The function $h$ is called \emph{locally Lipchitz} if, for any $M,N\in(0,\infty)$, there is a constant $\ell>0$ such that
$|h(t,x_1,y_1)-h(t,x_2,y_2)|\leq \ell (|x_1-x_2|+|y_1-y_2|)$ for all $t\in[0,T)$ and $(x_1,y_1)$, $(x_2,y_2)\in[0,M]\times [-N,N]$. In this case, $\ell$
is called a \emph{Lipchitz constant} of $h$ on $[0,T)\times[0,M]\times [-N,N]$.
\end{definition}

\begin{definition}
The function $ h$ is called \emph{locally bounded} if $ h$ is bounded on $[0,T)\times[0,M]\times [-N,N]$ for any $M,N\in(0,\infty)$. 
\end{definition}

\begin{definition}
The function $ h$ is called \emph{locally I-bounded}  if 
\begin{align*}
    \int_0^T \sup_{(x,y)\in[0,M]\times [-N,N]}|\mathcal{P}_t(h(t,x,y)\lambda(t))|^2|\md t<\infty
\end{align*}
 for any $M,N\in(0,\infty)$. The function $h$ is called \emph{I-bounded} if the above inequality holds when $[0, M] \times [-N, N]$ is replaced by $[0, \infty) \times \mathbb{R}$.
\end{definition}

\begin{remark}
  It is evident that if $h$ is (locally) bounded, then $h$ is (locally) I-bounded as well. Furthermore, if $\mathbb{U}_t \subset \tilde{\mathbb{U}}$ for some bounded set $\tilde{\mathbb{U}} \subset \mathbb{R}^d$ and for all $t\in[0,T)$, then $h$ must be I-bounded.
\end{remark}

For any $\phi\in L^2((t_1,t_2),\mR^d)$, we will frequently use the following notations:    \begin{align*}
        v_\phi(t_1,t_2)\triangleq \Vert \phi\Vert^2_{L^2((t_1,t_2),\mR^d)}=\int_{t_1}^{t_2}|\phi(s)|^2\md s, \quad y_{\phi}(t_1,t_2)\triangleq \int_{t_1}^{t_2}\phi^{\top}(s)\lambda(s)\md s.
    \end{align*}
    When $t_2=T$, we  simply write $v_\phi(t_1)$ and $y_\phi(t_1)$ for $v_\phi(t_1,T)$ and $y_\phi(t_1,T)$, respectively.
 
We begin with the following theorem on the uniqueness of the solution.

\begin{theorem}\label{thm:integraeq:uniqueness}
If $h$ is  locally Lipschitz, then,  for every $\tau\in[0,T)$, the equation 
\begin{align}\label{eq:integral:tau}
a(t)=\mathcal{P}_t\left(h\left(t,\sqrt{v_{a}(t)},y_a(t)\right)\lambda(t)\right), \quad t\in(\tau,T)
	\end{align}
has at most one solution $a\in L^2((\tau,T),\mR^d)$.
\end{theorem}
\begin{proof}
    Suppose that $a^{(i)}\in L^{2}((\tau,T),\mR^d)$, $i=1,2$,  are two solutions of \eqref{eq:integral:tau}.  
Let $$M_1\triangleq \max\left\{\sqrt{v_{a^{(1)}}(\tau)},\sqrt{v_{a^{(2)}}(\tau)}\right\}$$
and let $\ell_1$ be a Lipschitz constant of $h$  on
$$[0,T)\times[0,M_1]\times\left[-M_1\sqrt{v_{\lambda}(0)},M_1\sqrt{v_{\lambda}(0)}\right].$$  
Using the fact: $|\mathcal{P}_t(h_1)-\mathcal{P}_t(h_2)|\leq|h_1-h_2|$ for any $t\in(0,T)$ and $h_1, h_2\in \mR^d$, we have, for every $t_1\in[\tau,T)$ and every $t\in(t_1,T)$,
 we have 
	\begin{align*}
		&\left|a^{(1)}(t)-a^{(2)}(t)\right|
		\leq \ell_1 |\lambda(t)|\left[\left|\sqrt{v_{a^{(1)}}(t)}\!-\!\sqrt{v_{a^{(2)}}(t)}\right|\!+\!\left|y_{a^{(1)}}(t)\!-\!y_{a^{(2)}}(t)\right|\right]\\
		\leq& \ell_1 |\lambda(t)| \left[\sqrt{v_{a^{(1)}-a^{(2)}}(t)}+\sqrt{v_{\lambda}(t)}\cdot \sqrt{v_{a^{(1)}-a^{(2)}}(t)}\right] \leq \ell_1 |\lambda(t)|\left(1+\sqrt{v_{\lambda}(0)}\right) \sqrt{v_{a^{(1)}-a^{(2)}}(t_1)}.
	\end{align*}
    Therefore, we get
\begin{equation}\label{ineq:a12}
\begin{split}
		 \sqrt{v_{a^{(1)}-a^{(2)}}(t_1)}
		\leq \ell_1\sqrt{v_{\lambda}(t_1)}\left(1+\sqrt{v_{\lambda}(0)}\right)
		 \sqrt{v_{a^{(1)}-a^{(2)}}(t_1)},\quad\forall\, t_1\in[\tau,T).
\end{split}\end{equation}
Let $s_1=\inf\left\{t_1\in[\tau,T): \ell_1\sqrt{v_{\lambda}(t_1)}\left(1+\sqrt{v_{\lambda}(0)}\right)\leq \frac{1}{2}\right\}$. 
Then $a^{(1)}= a^{(2)}$ a.e. on $(s_1,T)$. If $s_1=\tau$, then $a^{(1)}= a^{(2)}$ a.e. on $(\tau,T)$. Otherwise, similar to \eqref{ineq:a12},  we have
\begin{align*}
		\sqrt{v_{a^{(1)}-a^{(2)}}(t_2,s_1)} \leq\ell_1\sqrt{v_{\lambda}(t_2,s_1)}\left(1+\sqrt{v_{\lambda}(0)}\right)\sqrt{v_{a^{(1)}-a^{(2)}}(t_2,s_1)},\quad \forall\,t_2\in[\tau,s_1). 
	\end{align*}
Let $s_2=\inf\left\{t_2\in[\tau,s_1): \ell_1\sqrt{v_{\lambda}(t_2,s_1)}\left(1+\sqrt{v_{\lambda}(0)}\right)\leq \frac{1}{2}\right\}$. Then $a^{(1)}= a^{(2)}$ a.e. on $(s_2,s_1)$. If $s_2=\tau$, then $a^{(1)}= a^{(2)}$ a.e. on $(\tau,T)$. Otherwise, we 
repeat such a procedure. Up to finitely many times, we can obtain that $a^{(1)}= a^{(2)}$ a.e. on $(\tau,T)$. In this case, 
$v_{a^{(1)}}= v_{a^{(2)}}$ and $y_{a^{(1)}}= y_{a^{(2)}}$ everywhere on $(\tau,T)$, which implies that $a^{(1)}= a^{(2)}$ everywhere on $(\tau,T)$.
\end{proof}

As regards the local existence of the solution, we have the following lemma.

\begin{lemma}\label{lma:integraeq}
If $h$ is locally I-bounded and locally Lipschitz, then there exists some $\tau^*\in[0,T)$ such that Equation \eqref{eq:integral:tau} has a solution $a\in L^{2}((\tau,T),\mR^d)$ for every $\tau\in[\tau^*,T)$.
\end{lemma}

\begin{proof}
For every $\tau\in[0,T)$, let 
$\mathcal{S}_\tau=\left\{a\in L^{2}((\tau,T),\mR^d): v_{a}(\tau) \leq 1\right\}$.
Obviously, $\mathcal{S}_\tau$ is a closed subset of $L^{2}((\tau,T),\mR^d)$. 
For every $a\in\mathcal{S}_\tau$, let
$$(\Ti a)_t\triangleq  \mathcal{P}_t\left(h\left(t,\sqrt{v_{a}(t)},y_{a}(t)\right)\lambda(t)\right) ,\quad t\in(\tau,T).$$ 

First, we show that $\Ti$ is a map from $\mathcal{S}_\tau$ to $\mathcal{S}_\tau$  for every $\tau\in[0,T)$ that is sufficiently close to $T$. Indeed, for every $a\in\mathcal{S}_\tau$, we have
    $\sqrt{v_{a}(t)}\leq 1\text{ and } \left|y_{a}(t)\right|\leq \sqrt{v_{\lambda}(0)}$, $\forall t\in[\tau,T)$.
Then
\begin{align*}
    v_{\Ti a}(\tau)\leq  \int_{\tau}^T \sup_{(x,y)\in[0,1]\times [-\sqrt{v_{\lambda}(0)},\sqrt{v_{\lambda}(0)}]}|\mathcal{P}_t(h(t,x,y)\lambda(t))|^2\md t.
\end{align*}
As $h$ is locally I-bounded, we can choose $\hat{\tau}\in[0,T)$ such that $ v_{\Ti a}(\hat{\tau})\leq 1$,  then $\Ti$ is a map from $\mathcal{S}_\tau$ to $\mathcal{S}_\tau$ for every $\tau\in[\hat\tau,T)$. 

Next, we show that there exists some $\tau^*\in[\hat\tau,T)$ such that $\Ti$ is a contraction on $\mathcal{S}_\tau$ for every $\tau\in[\tau^*,T)$. Indeed, let $\ell$ be a Lipschitz constant of $h$ on $[0,T)\times[0,1]\times [-\sqrt{v_{\lambda}(0)},\sqrt{v_{\lambda}(0)}]$. Using a similar way to \eqref{ineq:a12}, for any $\tau\in[\hat\tau,T)$ and $a^{(i)}\in \mathcal{S}_\tau$, $i=1,2$, 
we have
\begin{align*}
   \sqrt{v_{\Ti a^{(1)}-\Ti a^{(2)}}(\tau)}
   \leq
    \ell \sqrt{v_{\lambda}(\tau)}\left(1+\sqrt{v_{\lambda}(0)}\right) \sqrt{v_{a^{(1)}- a^{(2)}}(\tau)}.
\end{align*}
Choose a $\tau^*\in[\hat\tau,T)$ such that $\ell \sqrt{v_{\lambda}(\tau^*)}\left(1+\sqrt{v_{\lambda}(0)}\right)\leq \frac 12$, then
  $\Ti$ is a contraction on $\mathcal{S}_\tau$ for every $\tau\in[\tau^*,T)$. 
  
 Therefore, for every $\tau\in[\tau^*,T)$, $\Ti$ has a unique fixed point in $\mathcal{S}_\tau$, as such  \eqref{eq:integral:tau} has a solution in $L^{2}((\tau,T),\mR^d)$. 
\end{proof}
Assuming further that $h$ is I-bounded, we have the following global existence of the solution:
\begin{theorem}\label{thm:extend}
If $h$ is I-bounded and locally Lipschitz, then Equation \eqref{eq:integral:h} 
has a unique solution $a\in L^{2}((0,T),\mR^d)$.
\end{theorem}
\begin{proof}
Theorem~\ref{thm:integraeq:uniqueness} implies the following fact: if, for each $i=1,2$,  $a^{(i)}\in L^{2}((\tau_i,T),\mR^d)$ is a solution of \eqref{eq:integral:tau} for  $\tau=\tau_i\in[0,T)$ with  $0\le\tau_2\leq\tau_1<T$, then $a^{(1)}= a^{(2)}$ on $(\tau_1,T)$. 
  Let $(\tau,T)\subseteq(-\infty,T)$ be the maximal interval for the existence of solution\footnote{We have shown that $a^{(1)}$ can be extended to $a^{(2)}$ uniquely. Thus we can define the maximal interval for the existence of solution. 
  In addition, we define $\lambda(t)=\lambda(0)$, $\mathcal{P}_t=\mathcal{P}_0$ and $h(t,\cdot,\cdot)=h(0,\cdot,\cdot)$ for $t<0$, thus $\tau$ may belong to $[-\infty,0)$.
  }
  and $a$ be the solution on $(\tau,T)$. 
  Then we must have $\lim\limits_{t\downarrow \tau}v_a(t)=\infty$ if $\tau>-\infty$.  Suppose, for the sake of contradiction, that $\tau\in(-\infty,T)$ and $v_a(\tau)<\infty$.
  Consider the following equation for $\tilde{a}\in L^2((\tau-\epsilon_0,\tau),\mR^d)$:
  \begin{align}\label{eq:extend:h}
\tilde{a}(t)=\mathcal{P}_t\left(h\left(t,\sqrt{v_{\tilde{a}}(t,\tau)+v_a(\tau)},y_{\tilde{a}}(t,\tau)+y_a(\tau)\right)\lambda(t)\right), \quad t\in(\tau-\epsilon_0,\tau),
	\end{align}
    where $\epsilon_0>0$ is a constant to be determined.
   Using the same approach as in the proof of Lemma  \ref{lma:integraeq}, there is a $\epsilon_0>0$ such that \eqref{eq:extend:h} admits a solution in $L^2((\tau-\epsilon_0,\tau),\mR^d)$. Thus $a$ can be extended as the solution of \eqref{eq:integral:h} with $\tau$ replaced by $\tau-\epsilon_0$. This contradicts the assumption that $(\tau,T)$ is the maximal interval for the existence of solution.
  
  Assume that $h$ is I-bounded and $(\tau,T)$ is the maximal interval for the existence of solution.   We only need to show that $\tau<0$. Indeed, suppose on the contrary $\tau\geq 0$. Then  
  \begin{align*}
    \infty=\lim\limits_{t\downarrow \tau}v_a(t)\le  v_a(\tau)\leq \int_\tau^T \sup_{(x,y)\in[0,\infty)\times \mR}|\mathcal{P}_t(h(t,x,y)\lambda(t))|^2|\md t<\infty,
  \end{align*}
  which is impossible. 
 \end{proof}

\begin{remark}\label{rmk:h:bounded}
   Let $h$ be locally Lipschitz.  In the case  where $\lambda\in L^{\infty}((0,T),\mR^d)$,  $h$ is bounded, and no portfolio bounds are imposed, we, with Yuan, have established the global existence and uniqueness of the solution in $L^\infty((0,T),\mR^d)$; see  \citet*[Theorem 3.5]{liang2024dynamic}. In contrast, Theorem \ref{thm:extend} only requires $\lambda\in L^{2}((0,T),\mR^d)$  and $h$ is I-bounded to guarantee the existence and uniqueness of the solution in $L^{2}((0,T),\mR^d)$.  The I-boundedness of $h$ can be guaranteed by the boundedness of $h$ or the boundedness of $\{\mathbb{U}_t,t\in[0,1)\}$, which are easy to verify.
\end{remark}

\section{Portfolio selection  with time-inconsistent preferences}\label{sec:time-inconsistent}
We derive Equation \eqref{eq:integral:h} from   portfolio selection  with a class of time-inconsistent preferences in this section. 
Let $T>0$ be a finite time horizon and  $(\Omega,\mathcal{F},\mathbb{F},\mathbb{P})$ be a filtered complete probability space, where $\mathbb{F}=\{\mathcal{F}_t,0\leq t\leq T\}$ is the augmented natural filtration generated by a standard $d$-dimensional Brownian motion $\{B(t), 0\leq t\leq T\}$.
In addition, suppose $\F=\F_T$. The market consists of one risk-free asset (bond) and $d$ risky assets (stocks). For simplicity, we assume that the interest rate of the bond is zero. The dynamics of the stock price processes $S_{i}$, $i=1,\cdots,d$ are given by
\begin{align*}
	dS_{i}(t)=S_{i}(t)[\mu_{i}(t)\md t+\sigma_{i}(t) \md B(t)],\quad t\in[0,T],
\end{align*}
where the market coefficients $\mu:[0,T]\rightarrow\mathbb{R}^{d}$ and $\sigma:[0,T]\rightarrow\mathbb{R}^{d\times d}$ are right-continuous and deterministic, and $\sigma_i$ denotes the $i$-th row of $\sigma$. We always assume that $\int_0^T|\mu(t)|\md t+\sum_{i=1}^d\int_t^T|\sigma_i(t)|^2\md t<\infty$ and 
$\sigma(t)$ is invertible for every $t\in[0,T]$. Let $\lambda(t) = (\sigma(t))^{-1}\mu(t)$ be the \emph{market price of risk}. We also always assume that $\lambda\in L^2((0,T),\mR^d)$. 

For any $m\geq 1$ and $\mS \subset \mR^m$, $L^0(\mF,\mS)$ is the space of $\mS$-valued, $\mF$-progressively measurable processes. For each $t\in[0,T]$, $p\in[1,\infty]$, $L^p(\F_t,\mS)$ is the set of all $L^p$-integrable, $\mS$-valued, and $\F_t$-measurable random variables. 
A trading strategy is a process $\pi=\{\pi_t, t\in[0,T)\}\in L^0(\mF,\mR^d)$ such that $\int_0^T|\pi^\top_t\mu(t)|\md t+\int_0^T|\sigma^\top(t)\pi_t|^2\md t<\infty$ a.s., where $\pi_t$ represents the vector of portfolio weights determining the investment of wealth into the stocks at time $t$. The self-financing wealth process $\{W^{\pi}_t,0\leq t\leq T\}$ associated with a trading strategy $\pi$ satisfying the following stochastic differential equation (SDE):
\begin{equation}\label{wealthdynamic}
		\md W^\pi_t=W^\pi_t \pi^{\top}_t\mu(t)\md t+W^\pi_t \pi^{\top}_t\sigma(t)\md B(t),\quad	
		W^\pi_0=w_0>0.
\end{equation}
The time-$t$ performance of a trading strategy $\pi$ is $J(t,\pi)$.  We assume that $J(t,\cdot)$ can be (equivalently) represented as a functional of the conditional distribution of ${W^\pi_T\over W^\pi_t}$ given $\F_t$.  Some examples will be presented after Definition \ref{def:equilibrium}.

Let $\mathbb{U}\subset\mR^d$ be a convex and closed set with $\mathbf{0}\in\mathbb{U}$. 
 A trading strategy $\pi$ is called \emph{admissible}  if $\pi_t\in\mU$ a.s. and the time-$t$ functional $J(t,\pi)$ is well defined for all $t\in[0,T)$.
Let $\Pi$ denote the set of all admissible strategies. 

In general, as the optimal strategy at time $t$, which maximizes $J(t,\pi)$ over $\pi\in\Pi$, is not necessarily optimal at time $s>t$,  we consider the so-called equilibrium strategies. 
Hereafter, we always consider a fixed $\barpi\in\Pi$, which is a candidate equilibrium strategy. 
For any $t\in[0,T)$, $\epsilon\in(0,T-t)$ and $k\in  L^\infty(\F_t,\mathbb{U})$,   let $\barpi^{t,\epsilon,k}$ be given by
\[\barpi^{t,\epsilon,k}_s\triangleq
\left\{
\begin{aligned}
	&k, &s\in [t,t+\epsilon),\\
	&\barpi_s, &s\notin [t,t+\epsilon).
\end{aligned}
\right.
\]
The strategy $\barpi^{t,\epsilon,k}$ serves as a perturbation of $\barpi$. 
Following \cite*{hu2012time,Hu2017} and
\cite*{yan2019time}, we introduce the definition of equilibrium strategies as follows. 
\begin{definition}\label{def:equilibrium}
	$\barpi$ is called an \emph{equilibrium strategy} if, for any $t\in[0,T)$, we have
	\begin{equation}\label{limsup:J}
		\lim_{\epsilon\downarrow 0 }\esssup_{\epsilon_0\in(0,\epsilon)}\frac{J(t,\bar{\pi}^{t,\epsilon_0,k})-J(t,\bar{\pi})}{\epsilon_0}\leq 0\quad \text{a.s. for all }k\in  L^\infty(\F_t,\mU).
	\end{equation}
\end{definition}

 Here are two examples of preference functionals $J(t,\cdot)$ that depend only on the conditional distribution of ${W^\pi_T\over W^\pi_t}$ given $\F_t$, which will be further discussed in the next section.\footnote{Other examples include { the (endogenous habit, or mean) scaled MV preferences and (endogenous habit formation, or generalized) mean–standard-deviation preferences of \citet*{kryger2020optimal},}
the MV preference for the log returns of \cite*{Dai2021}, the constant-relative-risk-aversion (CRRA) betweenness preference discussed by \cite*{liang2023dynamic}, and the GDA preference discussed by  \cite*{liang2024dynamic}.
}
\begin{example} The MV preference functional with state-dependent risk aversion parameter of  \cite*{bjork2014mean} is given by
\begin{equation}\label{J:mv0}
J_0(t,\pi)=\mE_t\left[W^\pi_T\right]-\frac{\gamma}{2W^\pi_t}\textup{Var}_t\left[W^\pi_T\right]=W^\pi_tJ(t,\pi),
\end{equation}
where $\gamma>0$ is a risk aversion index and
\begin{equation}\label{J:mv}
J(t,\pi)=\mE_t\left[{W^\pi_T/ W^\pi_t}\right]-\frac{\gamma}{2}\textup{Var}_t\left[{W^\pi_T/ W^\pi_t}\right].
\end{equation}
Then, under Definition \ref{def:equilibrium}, an equilibrium strategy for \eqref{J:mv0} is an equilibrium strategy for \eqref{J:mv}, and vice versa.
\end{example}

\begin{example} The preference functional with random risk aversion of \cite{desmettre2023equilibrium} is given by 
\begin{equation}\label{J:rra0}
J_0(t,\pi)=\int_{(0,\infty)}\left(u^\gamma\right)^{-1}\left(\mathbb{E}_{t}\left[u^\gamma\left(W^\pi_T\right)\right]\right) \md \Gamma(\gamma)
=W^{\pi}_tJ(t,\pi),
\end{equation}
where $\Gamma$ is a distribution function on $(0,\infty)$, $u^\gamma$ are the CRRA utility functions given by
\begin{align*}
    u^{\gamma}(x)=\begin{cases}
        \frac{x^{1-\gamma}}{1-\gamma}, \quad \gamma>0,\gamma\neq 1, x>0,\\
        \log x, \quad \gamma=1, x>0,
    \end{cases}
\end{align*}
and
\begin{equation}\label{J:rra}
J(t,\pi)=\int_{(0,\infty)}\left(u^\gamma\right)^{-1}\left(\mathbb{E}_{t}\left[u^\gamma\left({W^\pi_T/ W^\pi_t}\right)\right]\right) \md \Gamma(\gamma).
\end{equation}
Then, under Definition \ref{def:equilibrium}, an equilibrium strategy for \eqref{J:rra0} is an equilibrium strategy for \eqref{J:rra}, and vice versa.
\end{example}

Because all the market coefficients are deterministic,  it is natural to conjecture that an equilibrium strategy $\barpi$ is in the form of
	\begin{align}\label{solution}
		&\barpi_s=(\sigma^{\top} (s))^{-1}a(s),\quad s\in[0,T),
	\end{align}
	where $a\in L^2([0,T),\mR^d)$ is a deterministic and right-continuous $\mR^d$-valued function such that  $a(s)\in \sigma^{\top} (s)\mU$  for all $s\in[0,T)$. Denote by $\mathcal{D}$ the trading strategies in the form of (\ref{solution}).	
	For any given $t \in [0, T)$, 
	it is easy to see that
\begin{align*}
{W^{\bar{\pi}}_T}/{W^{\bar{\pi}}_t}=\me^{\int_t^Ta^{\top}(s)\lambda(s)\md s-\frac12 \int_t^T|a(s)|^2\md s+\int_t^Ta^{\top}(s)\md B(s)}=\me^{y_a(t)-\frac12 v_a(t)+\int_t^Ta^{\top}(s)\md B(s)}.
\end{align*}
Therefore, 
\begin{align}\label{eq:WTWt}
{\text{given the condition $\F_t$,}}\quad{W^{\bar{\pi}}_T}/{W^{\bar{\pi}}_t}\sim\logn\left(y_a(t)-{\frac 12}v_a(t),v_a(t)\right).
 \end{align}
 As  $J(t,\cdot)$ can be represented as a functional of the conditional distribution of ${W^\pi_T\over W^\pi_t}$ given $\F_t$, we have\footnote{Here, we implicitly assume that $\mathcal{D}\subset \Pi$.}$J(t,\bar{\pi})=g(t,v_a(t),y_a(t))$
 for some $g: [0,T)\times[0,\infty)\times\mR\to\mR$. 
   {Using this function $g$, we can characterize the equilibrium strategies by an integral equation of $a$, as shown in the following theorem.}

\begin{theorem} \label{thm:equilibrium:chacr}
Assume that $\mathcal{D}\subset \Pi$, 
$g\in C^{0,1,1}([0,T)\times[0,\infty)\times\mR))$, and $g_v(t,v,y)<0$ for all $(t,v,y)\in[0,T)\times[0,\infty)\times\mR$. 
Then $\bar{\pi}=(\sigma^{\top})^{-1}a\in\mathcal{D}$ is an equilibrium strategy if and only if
		\begin{align}\label{eq:equilibrium:g}		a(t)=\mathcal{P}_t\left(-\frac{g_{y}(t,v_a(t),y_a(t))}{2g_v(t,v_a(t),y_a(t))}\lambda(t)\right), \quad t\in(0,T),
		\end{align}
		where $\mathcal{P}_t$ is 
 the projection from $\mR^d$ to $\mathbb{U}_t\triangleq \sigma^\top(t)\mathbb{U}$ defined in \eqref{eq:projection}.
	\end{theorem}
\begin{proof}
      As $\mathcal{D}\subset\Pi$, we know that      $J(t,\bar{\pi}^{t,\epsilon,k})$ is well defined for all $t\in[0,T)$, $k\in L^\infty(\F_t,\mU)$, and $\epsilon\in(0,T-t)$.
	Let $\barpi^{t,\epsilon,k}=(\sigma^{\top})^{-1}\tilde{a}$, then  
		\begin{align*}
			v_{\tilde{a}}(t)=v_a(t)+\int_t^{t+\epsilon}(|\sigma^{\top}(s)k|^2-|\sigma^{\top}(s)\bar{\pi}_s|^2)\md s, \quad y_{\tilde{a}}(t)=y_a(t)+\int_t^{t+\epsilon}\mu^{\top}(s)(k-\bar{\pi}_s)\md s.
		\end{align*}
		As $k$ is $\mathcal{F}_t$-measurable, we have $J(t,\barpi^{t,\epsilon,k})=g(t,v_{\tilde{a}}(t),y_{\tilde{a}}(t))$.
		Then 
		the condition \eqref{limsup:J} is equivalent to
		\begin{align}
			0\ge&\lim_{\epsilon\to0}\frac{g(t,v_{\tilde{a}}(t),y_{\tilde{a}}(t))\!-\!g(t,v_a(t),y_a(t))}{\epsilon}\nonumber\\
			=&g_v(t,v_a(t),y_a(t))|(\sigma^{\top}(t)k|^2-|\sigma^{\top}(t)\bar{\pi}_t|^2)+g_{y}(t,v_a(t),y_a(t))\mu^{\top}(t)(k-\bar{\pi}_t)\quad\forall k\in\mR^d,\nonumber
		\end{align}
        which is further equivalent to 
        \begin{align*}
            g_v(t,v_a(t),y_a(t))|(\tilde{k}|^2-|a(t)|^2)+g_{y}(t,v_a(t),y_a(t))\lambda^{\top}(t)(\tilde{k}-a(t))
			\le0\quad\forall \, \tilde{k}\in\sigma^{\top}(t)\mU.
        \end{align*}
		That is, 
		\begin{align*}
			a(t)=\arg\min_{\tilde{k}\in\sigma^{\top}(t)\mU}\left\{|\tilde{k}|^2+\frac{g_{y}(t,v_a(t),y_a(t))}{g_v(t,v_a(t),y_a(t))}\lambda^{\top}(t)\tilde{k}\right\}=\mathcal{P}_t\left(-\frac{g_{y}(t,v_a(t),y_a(t))}{2g_v(t,v_a(t),y_a(t))}\lambda(t)\right).
		\end{align*}
  Then by definition, $\bar{\pi}$ is an equilibrium strategy if and only if the equality in \eqref{eq:equilibrium:g} holds on $[0,T)$. The equality for $t=0$ can be removed from the integral equation because $a_0$ contributes nothing to the integrals and we can straightforwardly put $a(0)=\mathcal{P}_0\left(-\frac{g_{y}(0,v_a(0),y_a(0))}{2g_v(0,v_a(0),y_a(0))}\lambda(0)\right)$ to get the equality at $t=0$ once $\{a(t),t\in(0,T)\}$ solves \eqref{eq:equilibrium:g}.
	\end{proof}

 Let
 \begin{align}\label{eq:h}
		h(t,x,y)=-\frac{g_y(t,x^2,y)}{2g_v(t,x^2,y)}, \quad (t,x,y)\in[0,T)\times[0,\infty)\times\mR.
	\end{align}
	Then Equation \eqref{eq:equilibrium:g} is exactly  Equation \eqref{eq:integral:h}. Thus we complete deriving Equation \eqref{eq:integral:h}.

In the next section, we will apply the results on Equation \eqref{eq:integral:h} to some specific examples.
\section{Applications}\label{sec:applications}
\subsection{Dynamic Mean-Variance portfolio selection}
In this subsection, we examine the problem that was studied by \cite*{bjork2014mean},\footnote{The trading strategies in \cite*{bjork2014mean} is defined as the dollar amounts put into the risky assets. Their equilibrium strategy is formally a proportion strategy, but any perturbation of this strategy, which is a dollar-amount strategy, can lead to negative wealth. Therefore, as  \citet[ Page 3872]{he2021equilibrium} pointed out, the definition of the equilibrium strategy in \cite*{bjork2014mean} needs a modification such that their solution is indeed an equilibrium strategy. Our model, however, uses the proportion strategies to avoid this issue.} where the MV portfolio selection was investigated under the assumption that the agent’s risk aversion is inversely proportional to the wealth level. Specifically, the agent’s MV preference functional at time $t$ can be equivalently given by \eqref{J:mv}.
It is easy to see that $g$ and $h$ are independent of $t$ and, for all $(v,x,y)\in[0,\infty)^2\times\mR$, 
\begin{align}
    &g(v,y)=\me^y-\frac{\gamma}{2}\left(\me^{v+2y}-\me^{2y}\right), &&g_y(v,y)=\me^y-\gamma \left(\me^{v+2y}-\me^{2y}\right), \nonumber\\
&g_v(v,y)=-\frac{\gamma}{2}\me^{v+2y}, &&h(x,y)=\frac{1+\gamma\me^y}{\gamma \me^{y+x^2}}-1.\label{eq:h:mv}
	\end{align}

 The above function $h$ is obviously unbounded: $\lim\limits_{y\to-\infty}h(x,y)=\infty$ for all $x\in[0,\infty)$. Therefore,  we only have a local existence of the solution by Lemma \ref{lma:integraeq}. However, we can address this unboundedness issue by a prior estimation for the solutions in the next lemma.
\begin{lemma}\label{lma:prior}
Let function $h$ be given by \eqref{eq:h:mv}.
		Suppose that $a\in L^{2}((0,T),\mR^d)$ is a solution to Equation \eqref{eq:integral:h}. Then there exists a constant $M>0$, which depends only on $\lambda(\cdot)$, $T$, and $\gamma$, such that $h\left(\sqrt{v_a(t)},y_a(t)\right)\in[-M,1]$ for any $t\in[0,T]$.
	\end{lemma}
	\begin{proof}
		Obviously, $y_a$ satisfies the following  equation
		\begin{align*}
			y_a'(t)=-\lambda^{\top}(t)\mathcal{P}_t\left(h\left(\sqrt{v_a(t)},y_a(t)\right)\lambda(t)\right), \quad t\in[0,T), \quad y_a(T)=0.
		\end{align*}
		As $\textbf{0}\in\mU$, we have
		\begin{align*}
			\left|\mathcal{P}_t\left(h\left(\sqrt{v_a(t)},y_a(t)\right)\lambda(t)\right)\right|\leq \left|h\left(\sqrt{v_a(t)},y_a(t)\right)\right|\cdot |\lambda(t)|.
		\end{align*}
		Then
		\begin{align}\label{eq:y'}
			|y_a'(t)|\leq \left|h\left(\sqrt{v_a(t)},y_a(t)\right)\right|\cdot |\lambda(t)|^2\leq \left(2+\frac{1}{\gamma}\me^{-v_a(t)-y_a(t)}\right)|\lambda(t)|^2.
		\end{align}
		Next, observe that  
		\begin{align*}
			-v_a(t)-y_a(t)&\leq -v_a(t)+\sqrt{v_{\lambda}(0)}\sqrt{v_a(t)}=\sqrt{v_a(t)}(\sqrt{v_{\lambda}(0)}-\sqrt{v_a(t)})\leq \frac{{v_{\lambda}(0)}}{4}.
		\end{align*}
		Combining this with \eqref{eq:y'} yields 
		\begin{align*}
			|y'_a(t)|\leq \left(2+\frac{1}{\gamma}\me^{\frac{v_{\lambda}(0)}{4}}\right)|\lambda(t)|^2\triangleq C_1|\lambda(t)|^2.
		\end{align*}
		Then we  obtain  $|y_a(t)|\leq C_1 v_{\lambda}(0)$ for all $t\in[0,T]$.  Thus $h\left(\sqrt{v_a(t)},y_a(t)\right)\in[-M,1]$ for any $t\in[0,T]$ for some $M=M(T,\gamma,\lambda(\cdot))>0$.
	\end{proof}
    
 	Using Lemma \ref{lma:prior}, we get the following proposition.
	\begin{proposition}
 \label{prop:eq} Let function $h$ be given by \eqref{eq:h:mv}. Then 
 Equation \eqref{eq:integral:h} has a unique solution in $L^{2}((0,T),\mR^d)$. Therefore, 
		there is  a unique equilibrium in $\mathcal{D}$.
\end{proposition}	
\begin{proof}
    For the uniqueness, suppose that $a^{(1)}, a^{(2)} \in L^2((0,T),\mR^d)$ are two solutions of $\eqref{eq:integral:h}$ and let
    $M>0$ be the constant given by  
     Lemma \ref{lma:prior}. Then, $a^{(1)}, a^{(2)} \in L^2((0,T),\mR^d)$ are also two solutions of $\eqref{eq:integral:h}$ with
 $h$   replaced by $\tilde{h}=h\vee (-M)$. As $\tilde{h}$ is bounded, the uniqueness gives that $a^{(1)}=a^{(2)}$.
      For the existence, by Theorem \ref{thm:extend}, $\eqref{eq:integral:h}$ with $h$ replaced by $\tilde{h}=h\vee (-M)$ has a unique solution  $\tilde{a}\in L^2((0,T),\mR^d)$. We only need to show that $\tilde{a}$ is a solution of $\eqref{eq:integral:h}$. Indeed, as $|\tilde{h}|\leq|h|$, the same proof of Lemma \ref{lma:prior} yields
      \begin{align*}
          \tilde{h}(v_{\tilde{a}}(t),y_{\tilde{a}}(t))\in[-M,1], \quad t\in[0,T].
      \end{align*}
      Therefore $\tilde{h}(v_{\tilde{a}}(t),y_{\tilde{a}}(t))=h(v_{\tilde{a}}(t),y_{\tilde{a}}(t))$ and 
      \begin{align*}
          \tilde{a}(t)=\mathcal{P}_t\left(\tilde{h}(v_{\tilde{a}}(t),y_{\tilde{a}}(t))\right)=\mathcal{P}_t\left({h}(v_{\tilde{a}}(t),y_{\tilde{a}}(t))\right), \quad t\in(0,T).
      \end{align*}
   Thus $\tilde{a}$ is a solution of $\eqref{eq:integral:h}$. 
\end{proof}

\begin{remark}
In \cite*{bjork2014mean}, an iterative approach and the Arzelà-Ascoli theorem are used to study an integral equation similar to \eqref{eq:integral:h} with constant market coefficients and without portfolio constraints. In contrast, we allow for time-dependent and discontinuous market coefficients and portfolio constraints, our equilibrium strategies are only right-continuous. Their method does not apply to our case of discontinuous market coefficients because their use of the Arzelà-Ascoli theorem requires that the market coefficients are continuously differentiable.
\end{remark}	
\begin{remark}
    Consider the endogenous habit,  scaled MV preferences in Section 4.2.2 of \\ \citet*{kryger2020optimal}, where they do not provide the proof of 
    existence and uniqueness of  solution of their ODE system. 
We  define an equivalent preference functional 
$J(t,\pi)=\mE_t\left[{W^\pi_T/ W^\pi_t}\right]-\frac{\gamma}{2}\mE_t\left[\left({W^\pi_T/W^\pi_t}-k(t)\right)^2\right]$.
   Suppose  that $k(\cdot)$ is bounded and right-continuous. Using the same method as in Lemma \ref{lma:prior} and  Proposition \ref{prop:eq} , we can show that there is a unique equilibrium $\barpi=(\sigma^{\top})^{-1}a$ in $\mathcal{D}$, where $a$ satisfies 
    $a(t)=\left(\frac{1+\gamma k(t)}{\me^{y_a(t)+v_a(t)}}-1\right)\lambda(t)$, $t\in(0,T)$.
\end{remark}

\subsection{Equilibrium investment with random risk aversion}

In this section, we consider the portfolio selection problem with random risk aversion that was studied in \cite{desmettre2023equilibrium}. The time-$t$ performance functional is given by \eqref{J:rra}.
Let $\boldsymbol{R}$ be a random variable whose distribution function is $\Gamma$. 
Direct computation shows that $g$ is independent of $t$ and,  for any $(v,y)\in[0,\infty)\times \mR$, we have
 \begin{align*}
    &g(v,y)=g_y(v,y)=\int_{(0,\infty)}\me^{-\frac{1}{2}(\gamma v-2y)}\md \Gamma(\gamma)=\mE\left[e^{-\frac{1}{2}(\boldsymbol{R} v-2y)}\right],\\
     &g_v(v,y)=-\frac{1}{2}\int_{(0,\infty)}\gamma\me^{-\frac{1}{2}(\gamma v-2y)}\md \Gamma(\gamma)=-\frac{1}{2}\mE\left[\boldsymbol{R}\me^{-\frac{1}{2}(\boldsymbol{R} v-2y)}\right] .
 \end{align*}
Then, we know that $h$ is independent of $t$ and $y$, and has the following form:
\begin{align}\label{eq:randomaversion:h}
h(x)={\mE\left[\me^{-\frac{1}{2}\boldsymbol{R} x^2}\right]}\Big/{\mE\left[\boldsymbol{R}\me^{-\frac{1}{2}\boldsymbol{R} x^2}\right]}, \quad x\in[0,\infty).
\end{align}
To study the Lipschitz continuity of $h$, we consider the following function:
$f(x)=\mE\left[\me^{-\frac{1}{2}\boldsymbol{R} x^2}\right]$, $x\in[0,\infty)$. 
  The next lemma summarizes some well-known properties of $f$ and $h$; see,  \citet[Example 6.29]{Klenke2020}. 
\begin{lemma}\label{lma:randomaversion:h}
    We have the following conclusions:
    \begin{enumerate}    
        \item [(i)] $f,h \in C([0,\infty))\cap C^{\infty}((0,\infty))$.
        \item [(ii)] If $\mE\left[\boldsymbol{R}\right]<\infty$, then $f\in C^1([0,\infty))$ and $f'(x)=-\frac12\mE\left[\boldsymbol{R}\me^{-\frac{1}{2}\boldsymbol{R} x^2}\right].$
        \item [(iii)] If $\mE\left[\boldsymbol{R}^2\right]<\infty$, then $f\in C^2([0,\infty))$, $f''(x)=\frac14\mE\left[\boldsymbol{R}^2\me^{-\frac{1}{2}\boldsymbol{R} x^2}\right]$, and $h\in C^1([0,\infty))$.
   \end{enumerate}
    \end{lemma}
    
Combining Lemma \ref{lma:randomaversion:h} with Theorem \ref{thm:extend} leads to the following proposition.
    \begin{proposition}\label{prop:rra}
        If $\mE\left[\boldsymbol{R}^2\right]<\infty$ and $h$ is bounded, then there is a unique equilibrium in $\mathcal{D}$.
    \end{proposition}
  
    \begin{remark}
       In the case when $\Gamma$ is a two-point distribution, \cite{desmettre2023equilibrium} characterize the equilibrium by a three-dimensional ODE system\footnote{ In fact, with appropriate transformations, their ODE system can be converted into our integral equation. { Specifically, 
there are typos in their solution $g^\gamma(t)$ and their equation (37). The solution $g^\gamma$ to their equation (35) should be $$
g^\gamma(t)=\exp \left(-\int_t^T \frac{1-\gamma}{\gamma}\left(-r+0.5 \theta^2 \frac{1}{(\varepsilon(s, \Gamma))^2} \gamma-\theta^2 \frac{1}{\varepsilon(s, \Gamma)}\right) \md s\right)
$$ and hence their equation (37) for $\varepsilon(\cdot,\Gamma)$ should be
\begin{align}\nonumber
    \varepsilon(t,\Gamma)=&\frac{\int\left(\exp \left(-\int_t^T\left(-r+0.5 \theta^2 \frac{1}{(\varepsilon(s,\Gamma))^2} \gamma-\theta^2 \frac{1}{\varepsilon(s,\Gamma)}\right) \md s\right)\right) \gamma \md \Gamma(\gamma)}{\int\left(\exp \left(-\int_t^T\left(-r+0.5 \theta^2 \frac{1}{(\varepsilon(s,\Gamma))^2} \gamma-\theta^2 \frac{1}{\varepsilon(s,\Gamma)}\right) \md s\right)\right) \md \Gamma(\gamma)}\\
    =&\frac{\int\left(\exp \left(-\int_t^T\left(0.5 \theta^2 \frac{1}{(\varepsilon(s,\Gamma))^2} \gamma\right) \md s\right)\right) \gamma \md \Gamma(\gamma)}{\int\left(\exp \left(-\int_t^T\left(0.5 \theta^2 \frac{1}{(\varepsilon(s,\Gamma))^2} \gamma\right) \md s\right)\right) \md \Gamma(\gamma)}.\label{eq:epsilon:correct}
\end{align} Consider the case $r=0$ and let $
    a(t)\triangleq \frac{\alpha}{\sigma} \frac{1}{\varepsilon(t, \Gamma)}=\frac{\theta}{\varepsilon(t, \Gamma)}$.
       Then Equation \eqref{eq:epsilon:correct} is equivalent to our integral equation.} } without establishing the existence and uniqueness of the solution. Our approach, however, only requires that
$
\mE\left[\boldsymbol{R}^2\right]<\infty
$
and that $h$ is bounded to establish the existence and uniqueness of the solution. The boundedness of $h$ is fulfilled when the support of $\Gamma$ has a positive lower bound, that is, $\Gamma(\gamma_0)=0$ for some $\gamma_0>0$.  
    \end{remark}   
 \bibliographystyle{abbrvnat} 
 	\bibliography{sample.bib}	
\end{document}